\documentclass[letterpaper,journal]{IEEEtran}
\usepackage{cite}
\usepackage{amsmath,amssymb,amsfonts}
\usepackage{booktabs}
\newtheorem{definition}{Definition}
\newtheorem{proposition}{Proposition}
\newtheorem{remark}{Remark}
\usepackage{algorithm}
\usepackage{algorithmic}

\usepackage{graphicx}
\usepackage[caption=false,font=footnotesize]{subfig}
\usepackage{textcomp}
\usepackage{url}
\usepackage{multirow}
\usepackage{xcolor}
\usepackage{hyperref}

\begin{document}

\title{SEM-RAG: Structure-Preserving Multimodal Graph Compilation and Entropy-Guided Retrieval for Telecommunication Standards}

\author{Yuzhi~Yang, Lina~Bariah, Yuhuan~Lu, Hang~Zou, and~M\'erouane~Debbah
\thanks{Y. Yang, L. Bariah, Y. Lu, H. Zou, and M. Debbah are with Research Institute for Digital Future, Khalifa University, 127788 Abu Dhabi, UAE (e-mails: \{firstname.lastname\}@ku.ac.ae).}
}

\markboth{}{Yang \MakeLowercase{\textit{et al.}}: SEM-RAG}

\maketitle

\begin{abstract}
Telecommunication standards pose a unique challenge for retrieval systems, where accuracy depends on semantic relevance as well as on preserving the structural logic embedded in the documents, including structured relationships embedded in tables, conditions, and formulas. When these elements are flattened into text, critical dependencies are lost, leading to unreliable retrieval. In this paper, we present SEM-RAG, an end-to-end retrieval framework built around two design choices. First, a layout-aware compiler converts text, tables, and formulas into typed graph primitives. Each table cell is linked to its row headers, column headers, predicates, and source coordinates, while each formula is converted into an operator graph tied to nearby symbol definitions. Second, the compiled graph is compressed with Structural Entropy Minimization (SEM), which avoids LLM-based bottom-up clustering during indexing. A Jensen-Shannon alignment layer and a lightweight query controller serve as supporting retrieval components that map user queries to the right subgraphs, while keeping online cost stable. Experiments on TeleQnA, TSpec-LLM, SPEC5G, and ORAN-Bench-13K show that SEM-RAG improves performance on table-heavy and formula-heavy questions, reaches 94.1\% accuracy on TeleQnA and 93.8\% on ORAN-Bench-13K, and cuts indexing-time token usage by a wide margin relative to standard GraphRAG. These results indicate that structure-preserving compilation is a practical requirement for retrieval over telecom specifications, not merely an optional preprocessing step.
\end{abstract}


\section{Introduction}
\label{sec:introduction}
The rollout of 5G-Advanced and the early design of 6G systems are driving telecom operations toward increasingly automated and data-intensive workflows \cite{bariah2024next, zhou2024large, zou2026nree}. Yet, in practice, network behavior remains governed by extensive standards documents maintained by 3GPP and the O-RAN Alliance \cite{nikbakht2024tspec, gajjar2024oranbench}. These documents are fundamentally different from conventional text. They interleave hierarchical clauses with dense parameter tables, option matrices, notation-heavy formulas, and tightly cross-referenced definitions. As a result, much of their operational meaning is encoded not in text alone, but in layout, conditional structure, and local dependencies.

This creates a difficult setting for standalone large language models (LLMs). Even strong general-purpose models struggle to reproduce exact configuration thresholds, release-specific option flags, or clause-level exceptions from parametric memory alone. Results on benchmarks such as TeleQnA show that unsupported answers remain common even for advanced models \cite{maatouk2023teleqna}, and similar observations have been reported by recent telecom-oriented evaluations \cite{gsma2025open}. In this domain, a minor factual error can lead to a wrong capability decision, a misconfigured slice, or an avoidable service failure.

\begin{figure*}[t]
    \centering
    \includegraphics[width=0.85\linewidth]{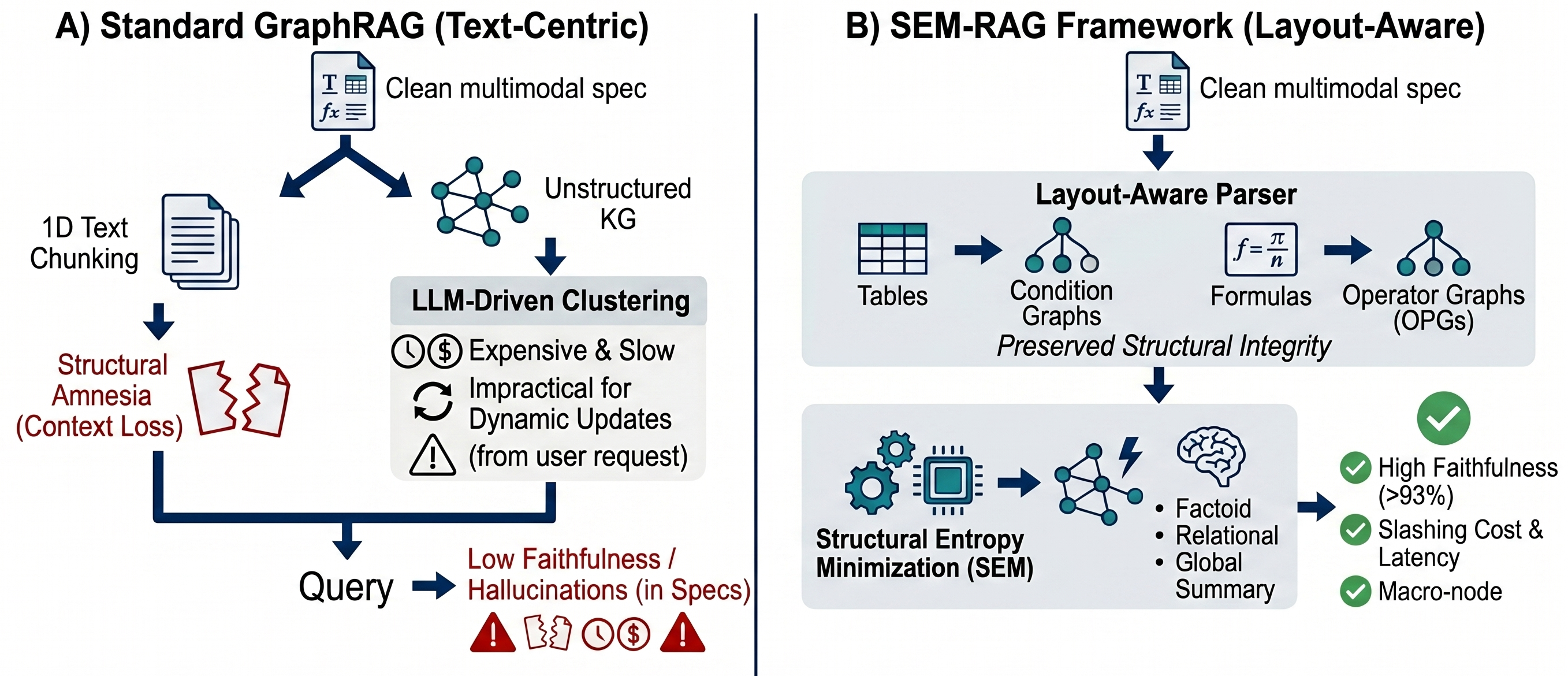}
    \caption{Conceptual comparison between a standard text-centric GraphRAG pipeline and SEM-RAG. In the left panel, standard GraphRAG starts from one-dimensional text chunking, loses structural relations in tables and formulas, and then relies on LLM-driven clustering over an unstructured knowledge graph. In the right panel, SEM-RAG first applies a layout-aware parser that converts tables into condition graphs and formulas into operator graphs, then uses Structural Entropy Minimization (SEM) to build a compact hierarchy that supports factoid, relational, and global-summary retrieval.}
    \label{fig:intro}
\end{figure*}

For this reason, telecom assistants increasingly rely on retrieval-augmented generation (RAG) to ground responses in the underlying specifications \cite{gao2023retrieval}. Yet the standard RAG pipeline still inherits a simplifying assumption from open-domain question answering, where documents are reduced to one-dimensional spans and then ranked primarily by semantic similarity. This approximation is often acceptable for descriptive paragraphs, but it is much less adequate for standards whose meaning depends on row–column bindings, applicability notes, and symbolic computation. Once a modulation table, service bitmap, or physical-layer equation is flattened into plain text, the link between headers, cells, conditions, symbols, and numerical values is either weakened or entirely lost. We refer to this failure mode as \emph{structural amnesia}, where the system keeps the words but not the structural relations that make those words operational.

GraphRAG methods partly address this problem by introducing explicit entities and relations \cite{edge2024global, pan2024unifying}. They are especially useful when a question requires evidence drawn from several connected clauses rather than from a single passage. Nevertheless, current graph-based pipelines are limited in two ways that matter in telecom standards. First, many still begin with text-first parsing, so the graph is constructed only after the original table and formula structure has already been discarded. Second, the hierarchy is often produced through repeated LLM calls for community discovery or bottom-up summarization \cite{edge2024global}. As illustrated in the left panel of Fig.~\ref{fig:intro}, this design creates a double bottleneck, where the initial graph already reflects context loss from one-dimensional chunking, and the subsequent LLM-driven clustering is costly, slow, and difficult to update when the corpus changes.

These observations reveal a gap that is specific to telecom standards. Efficient text-centric pipelines are easy to deploy, but they break the executable structure of tables and formulas at ingestion time. Graph-based pipelines improve relational retrieval, but they often start from already flattened content and then add an expensive LLM-driven hierarchy on top. Table~\ref{tab:positioning} summarizes this gap and clarifies where SEM-RAG is positioned. What remains missing is a framework that preserves structural relations from the start and still supports practical indexing and updates.

\begin{table*}[t]
\caption{Positioning of SEM-RAG against representative retrieval paradigms for telecom standards.}
\label{tab:positioning}
\centering
\small
\begin{tabular}{@{}p{3.0cm}p{4.2cm}p{7.3cm}@{}}
\toprule
\textbf{Method family} & \textbf{Main strength} & \textbf{Main limitation in telecom standards} \\
\midrule
Naive RAG / passage retrieval & Simple pipeline and low indexing overhead & Flattens tables and formulas into text chunks, so row-to-cell dependencies and operator structure are not preserved \\
Hierarchical text RAG & Improves long-context organization over flat chunking & Still builds hierarchy over textual summaries rather than over executable table or formula structure \\
GraphRAG & Supports relational retrieval and multi-hop traversal & Often inherits text-first parsing and relies on costly LLM-based clustering to build the hierarchy \\
SEM-RAG & Preserves table and formula structure during ingestion and builds the hierarchy with SEM instead of recursive LLM summarization & Requires a richer parser and typed graph construction during indexing \\
\bottomrule
\end{tabular}
\end{table*}

In telecom standards, a key limitation of existing retrieval pipelines lies in how documents are ingested rather than how they are queried. When tables and formulas are flattened into text, row–column bindings, conditional logic, and operator dependencies are no longer explicitly represented. As a result, later retrieval stages operate on incomplete structure and often fail to recover the correct operational constraints. SEM-RAG addresses this issue by preserving structure at the point of compilation. As shown in the right panel of Fig.~\ref{fig:intro}, the framework first applies a layout-aware parser that converts tables into \emph{condition graphs} and formulas into \emph{operator graphs}, while retaining provenance to the source clauses. The resulting multimodal graph is then compressed using Structural Entropy Minimization (SEM), which constructs a hierarchical index without relying on LLM-based clustering. Query-time retrieval operates directly over this compiled structure through factoid, relational, and global retrieval modes. A Jensen-Shannon alignment layer and a lightweight MLP router support this process, but they are included to make retrieval over the compiled graph practical rather than to serve as the central novelty of the paper.

This design has two immediate consequences. First, retrieval operates on graph primitives that preserve structural integrity instead of on flattened text fragments. Second, the indexing pipeline no longer depends on recursive LLM summarization of intermediate communities, which substantially reduces token cost and re-indexing latency. Hence, SEM-RAG is not introduced as a retrieval framework in which the parser and the graph index are co-designed for standards documents whose logic is encoded in layout.

The contributions of this work are fourfold:
\begin{enumerate}
\item We formulate a structure-preserving compiler for telecom standards that converts hierarchical text, relational tables, and mathematical expressions into typed graph primitives. In particular, the compiler produces hierarchical paragraph primitives \(\mathcal{T}_i\), conditional-logic table subgraphs \(\mathcal{C}_i\), and formula graph primitives \(\mathcal{F}_i\), while preserving provenance back to the source clauses. This addresses the representation gap left by text-only ingestion.
\item We introduce an entropy-guided indexing procedure that compresses the compiled graph without LLM-based bottom-up clustering. The resulting hierarchy is used as a retrieval index, while LLM calls are limited to a small number of top-level summaries. This addresses the indexing-cost gap of LLM-centered GraphRAG.
\item We incorporate a Jensen-Shannon divergence (JSD)-based alignment layer and a lightweight query controller as query-time support components so that natural-language questions can be matched to the appropriate compiled subgraphs without traversing the entire graph. These components improve queryability of the compiled index under practical latency constraints.
\item We provide a broad empirical study on telecom benchmarks, including table-heavy and formula-heavy slices, multi-hop reasoning tests, indexing-cost comparisons, and cross-LLM robustness analyses, and we state explicitly which dataset supports each diagnostic.
\end{enumerate}

\section{Related Work}
\label{sec:related_work}

\subsection{Large Language Models in Telecommunications}

Recent work on LLMs in telecommunications shows that raw language understanding is not enough for standards-intensive network operations. Cellular systems are governed by documents that are unusually dense in abbreviations, release-specific conditions, parameter tables, and normative constraints. As a result, even very capable foundation models often produce answers that sound plausible while missing a release dependency, misreading a threshold, or overlooking a table condition that changes the meaning of an otherwise familiar term \cite{zhou2024large, bariah2023globecom}. This mismatch is particularly pronounced in tasks tied to radio configuration, protocol compliance, security procedures, and O-RAN interoperability, where a small factual deviation can make the entire answer invalid.

Early studies in this direction mainly examined whether general-purpose or telecom-adapted LLMs could internalize standards knowledge through prompting or fine-tuning \cite{zou2025telecomgpt}. Benchmarks such as TeleQnA and ORAN-Bench-13K showed that these models can capture general conceptual knowledge, but they still struggle when the answer depends on exact wording, cross-section constraints, or specification artifacts that are difficult to memorize parametrically \cite{maatouk2023teleqna, gajjar2024oranbench}. This line of work clarifies both the opportunity and the limitation, showing that telecom knowledge is rich enough to benefit from LLM reasoning, yet sufficiently constrained that ungrounded generation is unreliable.

This observation led naturally to retrieval-based telecom assistants. Systems such as TelcoAI and Telco-RAG advanced from pure generation toward grounding on external corpora, typically by retrieving semantically similar chunks from specification text and passing them to the LLM as evidence \cite{conger2025telcoai, bornea2025telcorag}. These frameworks improved the accuracy compared with zero-shot generation, but they maintain a simplifying assumption inherited from general-domain RAG, namely that documents can be treated primarily as text. That assumption is significantly more limiting in telecom specifications than in general natural language text. Important evidence is frequently distributed across structured elements such as tables, enumerated conditions, footnotes, and equations, where meaning depends on layout and local structure rather than lexical similarity alone.

More recent work has begun to acknowledge that telecom QA requires more than better embeddings or larger context windows. Studies on retrieval design for technical corpora suggest that the primary failure mode is not missing vocabulary but missing structure, where the system retrieves related text yet fails to preserve the internal logic that makes the text operationally meaningful \cite{zhou2025llm}. Curriculum-style retrieval and modality-aware alignment strategies follow a similar direction, aiming to expose models to more complex structural evidence in a more organized manner \cite{sun2025curriculum, chu2025A}. Our work follows this broader shift, but targets a more specific bottleneck. For telecom standards, the bottleneck lies not only in semantic retrieval but in how multimodal specification content is compiled into representations that retain executable relationships. SEM-RAG is designed to address this requirement.

\subsection{Evolution of Graph-Enhanced Retrieval Architectures}

Graph-enhanced retrieval was introduced to address the limitations of passage-level dense retrieval. When evidence is distributed across entities, references, and long-range dependencies, a flat collection of chunks struggles to support multi-hop reasoning. GraphRAG-style systems handle this by representing documents as graphs of interconnected nodes and edges, allowing retrieval to move beyond local lexical similarity and toward explicit relational traversal \cite{pan2024unifying}. This is particularly relevant for technical domains, where many questions require tracing dependencies rather than extracting a single sentence.


Microsoft's GraphRAG system realized this paradigm by combining graph construction with hierarchical clustering and summary generation \cite{edge2024global}. Its key insight is that a graph can support both local evidence retrieval and higher-level corpus navigation. However, the standard implementation also exposes a systems limitation, as graph scale increases, community detection, recursive summarization, and repeated LLM calls become costly to construct and difficult to maintain. This trade-off is manageable for relatively static corpora, but it becomes restrictive in domains such as telecommunications, where specifications evolve across releases and practical deployments require frequent re-indexing as the knowledge base grows.


Subsequent work focuses on improving efficiency and traversal behavior in graph-based retrieval. Some approaches refine graph construction through concept-aware entity alignment and improved relation extraction, reducing spurious edges and improving traversal fidelity \cite{hu2026knowledge}. Others modify retrieval policies, where GraphFlow and ReGraphRAG adjust how the system traverses the graph or reconstructs task-relevant subgraphs at query time \cite{yu2025graphflow, kim2025regraphrag}. Lightweight variants such as LightRAG and PathRAG reduce overhead by limiting hierarchy depth, simplifying graph structure, or constraining traversal scope \cite{guo2024lightrag, chen2026pathrag}. In contrast, HyperGraphRAG and related extensions extend the graph representation itself, arguing that pairwise edges are insufficient to capture higher-order relational dependencies \cite{luo2025hypergraphrag}.

These studies show that graph-based retrieval is not a single design point but a spectrum of trade-offs among structural expressiveness, indexing cost, traversal depth, and runtime efficiency \cite{xiang2025when}. However, a key valid assumption is that the graph is typically built from text-centric extraction, and structural organization is often delegated to LLM-heavy procedures. For telecom specifications, this leads to two practical gaps. First, if the source representation has already lost table structure or equation semantics, downstream graph reasoning can only partially recover the missing dependencies. Second, if graph hierarchy depends heavily on repeated LLM clustering, maintaining and updating the index becomes computationally expensive in evolving corpora. SEM-RAG is aimed to maintain the graph-based view of retrieval, but shifts the emphasis from semantic graph decoration to structure-preserving graph compilation, and from LLM-based hierarchy construction to structural-entropy-driven compression.


\subsection{Layout-Aware Parsing and Multimodal Understanding}

A parallel line of research comes from Document AI, where the central concern is that document meaning is often inseparable from layout. In forms, manuals, scientific articles, and standards documents, two pieces of text with similar wording may play very different roles depending on whether they appear in a title, a row header, a footnote, or a formula. Layout-aware parsing methods were developed precisely to avoid destroying these details during preprocessing \cite{sourati2025ladrag}. Rather than extracting text as a linear stream, these methods use page geometry, reading order, block segmentation, and visual context to preserve the local structure of the source document.

Recent retrieval-oriented systems have extended these ideas into multimodal RAG. For example, SuperRAG and Infinity Parser demonstrate that preserving bounding boxes, regions, and hierarchical layout can materially improve document understanding when answers depend on visually organized evidence rather than ordinary prose \cite{yang2025superrag, wang2025infinity}. MAGMaR and related multimodal frameworks extend this idea further by attempting to align text, tables, charts, and other visual elements within a shared reasoning space \cite{drushchak2025multimodal}. These results highlight that the ingestion stage defines the information available to the retriever, rather than acting as a simple preprocessing step.

However, most layout-aware systems are still optimized for representation fidelity rather than executable semantics. This distinction is critical in technical standards. In telecom documents, the goal is not just to preserve where a table cell appears, but what logical condition that cell encodes. A matrix may encode dependencies between a service option, a capability flag, and a release-specific state. An equation may define an operational quantity whose interpretation depends on variable roles and operator ordering. If these objects are only serialized as text with coordinates attached, the LLM is forced to infer the operational logic during answer generation.

This is where recent work on mathematical structure becomes particularly relevant. Studies on formula understanding increasingly argue that equations should be parsed into structured computational objects rather than treated as opaque strings \cite{mansouri2025math}. Operator-level representations can separate variables, constants, and functional dependencies in a way that is much more useful for grounded reasoning. Related efforts on localized fragment augmentation similarly suggest that targeted structural decomposition can improve model behavior on technical content \cite{lvov2026enhancing}. SEM-RAG builds on this concept, but specializes it for telecom standards by treating tables and formulas as graph primitives to be compiled, not just preserved. In other words, our parser is not intended to be a better OCR layer, it is intended to produce a representation that supports later retrieval and reasoning in a form closer to the original operational logic of the specification.

The second technical component of this work comes from graph learning and network compression, where Structural Entropy (SE) provides an information-theoretic measure of graph organization. In the formulation of the authors in \cite{li2016structural}, structural entropy quantifies the expected coding cost of a random walk on a graph, thereby capturing the remaining uncertainty in its connectivity pattern. This perspective recasts graph partitioning as the problem of constructing a hierarchical organization that minimizes structural uncertainty, rather than as purely semantic clustering.

Later work has shown that SE is useful beyond its original formulation. In hierarchical graph analysis, minimizing structural entropy tends to expose dense local regularities while separating weaker cross-community connections \cite{cao2024hierarchical}. This makes it relevant to compression, community discovery, and efficient graph indexing. More recent models such as GSE-TLS, HypCSE, and T-Retriever demonstrate that entropy-guided grouping can produce competitive representations without relying on the heavy optimization or repeated language-model supervision that many contemporary retrieval systems use \cite{peng2025gse, zeng2026hyperbolic, wei2026tretriever}. The common lesson is that graph organization can often be improved substantially by exploiting topology itself, rather than repeatedly asking an LLM to explain what the graph means.

For our setting, the importance of SE is mainly methodological. Telecom standards give rise to large heterogeneous graphs with many local regularities, including repeated table schema, recurrent protocol entities, and formula fragments that connect tightly within a subtopic, while remaining loosely coupled to the broader corpus. A hierarchy induced through structural entropy can exploit these regularities in a deterministic way. This does not eliminate semantics from the system, semantic labeling is still useful at higher levels. But it changes where semantics are needed. Instead of using LLM calls to discover the hierarchy itself, SEM-RAG uses SE to build the hierarchy first and reserves language modeling for the smaller task of interpreting the resulting macro-structure.

This distinction is important because it sharpens the contribution boundary of the present work. Our aim is not to propose a new general-purpose SE theory, nor to claim that structural entropy alone solves retrieval. Rather, we use SE as a practical systems mechanism for constructing a scalable hierarchy over a parser-generated multimodal graph. In that sense, structural entropy is the compression engine that makes the architecture operational, while the parser remains the component that determines whether the graph faithfully reflects telecom document structure in the first place. This division of roles is central to the design of SEM-RAG and motivates the methodological choices in the next section.

\section{Problem Formulation and Design Rationale}
\label{sec:problem_formulation}

\subsection{Mathematical Representation of the Target Corpus}
We consider a corpus of telecom standards \(\mathcal{D} = \{D_1, D_2, \dots, D_N\}\), where each document $D_i$ contains three source modalities,
\begin{equation}
    D_i = \{X_i^{\text{text}}, X_i^{\text{tab}}, X_i^{\text{eq}}\},
\end{equation}
where \(X_i^{\text{text}}\) denotes the textual clauses, \(X_i^{\text{tab}}\) denotes the relational tables with row and column spans, and \(X_i^{\text{eq}}\) denotes the displayed mathematical expressions. Figures and sequence charts are retained as provenance metadata when present, but they are not the primary target of the compiler studied here.

The layout-aware compiler transforms \(D_i\) into three compiled components,
\begin{equation}
    \Phi(D_i) = \{\mathcal{T}_i, \mathcal{C}_i, \mathcal{F}_i\},
\end{equation}
where \(\mathcal{T}_i\) denotes hierarchical paragraph primitives, \(\mathcal{C}_i\) denotes conditional-logic table subgraphs, and \(\mathcal{F}_i\) denotes formula graph primitives. The compiled outputs from all documents are then merged into the base multimodal graph
\begin{equation}
    \mathcal{G}_0 = (\mathcal{V}, \mathcal{E}, \tau, \rho, \psi),
\end{equation}
where \(\tau : \mathcal{V} \rightarrow \mathcal{Y}\) assigns node types, \(\rho : \mathcal{E} \rightarrow \mathcal{R}\) assigns relation labels, and \(\psi\) stores attributes such as source coordinates, clause identifiers, units, and release tags.

Equations (1) to (3) should be read in sequence: (1) defines the raw document modalities, (2) defines their compilation into typed primitives, and (3) defines the merged graph that is indexed later.

\subsection{Executable Graph Primitives}
\begin{definition}[Executable graph primitive]
An executable graph primitive is a typed subgraph whose meaning can be recovered by binding predicates or traversing operator dependencies, without reconstructing the original two-dimensional document layout at query time.
\end{definition}

For a table instance \(b \in X_i^{\text{tab}}\) with row headers \(\mathcal{H}^{r} = \{h_i^{r}\}_{i=1}^{m}\), column headers \(\mathcal{H}^{c} = \{h_j^{c}\}_{j=1}^{n}\), data cells \(\mathcal{X} = \{x_{ij}\}\), and optional predicates \(\mathcal{P} = \{p_{ij}\}\) extracted from footnotes or applicability notes, the compiler builds a conditional-logic subgraph
\begin{equation}
    \mathcal{C}(b) = (\mathcal{V}_{b}, \mathcal{E}_{b}),
\end{equation}
with
\begin{equation}
    \mathcal{V}_{b} = \mathcal{H}^{r} \cup \mathcal{H}^{c} \cup \mathcal{X} \cup \mathcal{P},
\end{equation}
and
\begin{equation}
\begin{aligned}
    \mathcal{E}_{b} = {} & \{(h_i^{r}, x_{ij})\} \cup \{(h_j^{c}, x_{ij})\} \\
    & \cup \{(p_{ij}, x_{ij})\} \cup \mathcal{E}_{\text{src}},
\end{aligned}
\end{equation}
where \(\mathcal{E}_{\text{src}}\) links each cell back to its clause and page coordinates. At the document level, the compiled table output is \(\mathcal{C}_i = \bigcup_{b \in X_i^{\text{tab}}} \mathcal{C}(b)\). A lookup over the table then becomes a constrained traversal over headers and predicates, rather than string matching over flattened rows.

Equations (4) to (6) therefore specify how one table instance is compiled into a traversable condition structure.

Given an equation instance \(e \in X_i^{\text{eq}}\) and its abstract syntax tree \(A_e\), the compiler builds a formula graph
\begin{equation}
    \mathcal{F}(e) = (\mathcal{U}_{e} \cup \mathcal{K}_{e} \cup \mathcal{O}_{e}, \mathcal{E}^{\text{op}}_{e} \cup \mathcal{E}^{\text{def}}_{e}),
\end{equation}
where \(\mathcal{U}_{e}\) are variable nodes, \(\mathcal{K}_{e}\) are constant nodes, and \(\mathcal{O}_{e}\) are operator nodes. The edge set \(\mathcal{E}^{\text{op}}_{e}\) preserves operand order and precedence, while \(\mathcal{E}^{\text{def}}_{e}\) connects symbols to their local textual definitions. At the document level, the compiled formula output is \(\mathcal{F}_i = \bigcup_{e \in X_i^{\text{eq}}} \mathcal{F}(e)\). This gives the retriever access to sub-expressions and symbol semantics instead of only the surface LaTeX string.

Equation (7) plays the same role for formulas, namely turning a displayed expression into a retrievable primitive with explicit operator structure.

\subsection{Structural Amnesia as the Main Failure Mode}
Standard text chunking induces a lossy mapping \(\Pi : D_i \mapsto \{s_k\}_{k=1}^{M_i}\) that preserves token order inside each span but ignores two-dimensional adjacency, header inheritance, and operator precedence. For tables and formulas, this removes information required to recover the original selection or computation logic. We call this loss \emph{structural amnesia}. In telecom standards, a question can often be answered only after identifying the correct row, the correct column, the relevant release condition, and the applicable formula scope; once these bindings are dropped, later retrieval stages have little chance to recover them reliably.

\subsection{Why LLM-Driven Hierarchical Indexing Becomes Expensive}
Let \(|C_l|\) denote the number of communities at hierarchy level \(l\) and let \(T_{\text{prompt}}\) be the average token budget required to summarize one community. If every level requires LLM summaries, the indexing-time token usage grows as
\begin{equation}
    \mathrm{Tok}_{\text{LLM}} \propto \sum_{l=1}^{L} |C_l| T_{\text{prompt}}.
\end{equation}
This becomes costly in technical corpora because the graph must be refreshed as standards evolve. SEM-RAG targets this bottleneck by performing hierarchy construction through graph optimization first and delaying LLM usage until the final macro-node summaries are needed.

\begin{figure*}[t]
    \centering
    \includegraphics[width=\linewidth]{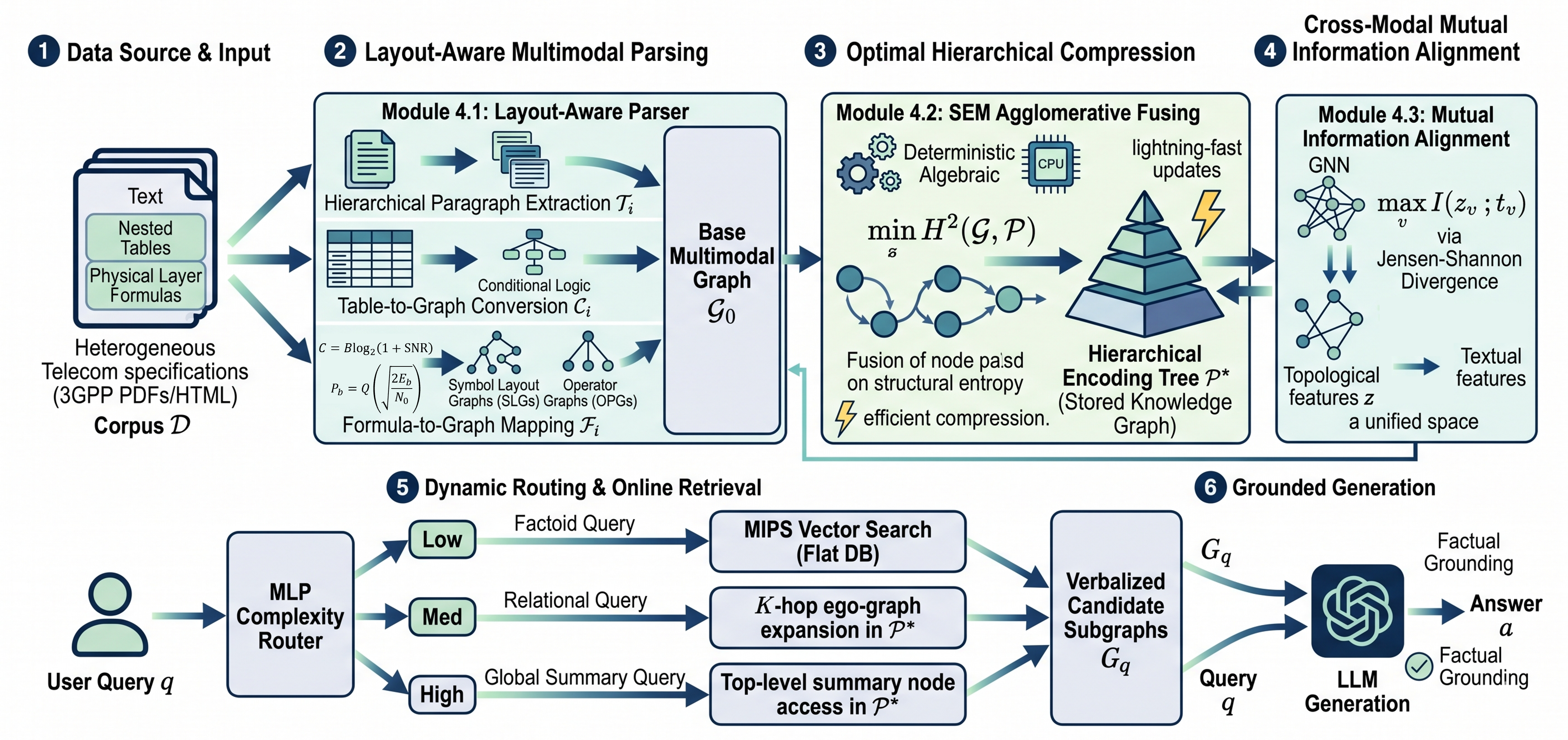}
    \caption{Overview of SEM-RAG. Module 4.1 extracts hierarchical paragraph primitives \(\mathcal{T}_i\), conditional-logic table subgraphs \(\mathcal{C}_i\), and formula graphs \(\mathcal{F}_i\), which are merged into the base multimodal graph \(\mathcal{G}_0\). Module 4.2 compresses \(\mathcal{G}_0\) into the hierarchical encoding tree \(\mathcal{P}^*\). Module 4.3 aligns topological features \(z_v\) and textual features \(t_v\). At query time, the complexity router selects a retrieval path and produces candidate subgraph \(G_q\) for grounded generation.}
    \label{fig:overall}
\end{figure*}

\section{Methodology: The SEM-RAG Architecture}
\label{sec:methodology}
SEM-RAG operates in an offline stage and an online stage, as shown in Figure~\ref{fig:overall}. The offline stage is the core of the method, where source documents are compiled into typed graph primitives, compressed into a hierarchy with SEM, and aligned with textual embeddings. The online stage then uses a lightweight controller to choose a retrieval mode and assemble evidence for grounded generation.

\subsection{SEM-RAG at a Glance}
At a high level, SEM-RAG has four layers. The first layer is a document-compilation layer that turns clauses, tables, and formulas into typed graph primitives with provenance. The second layer is an indexing layer that compresses the compiled graph into a hierarchy using SEM rather than recursive LLM summarization. The third layer is an alignment layer that reduces the gap between textual queries and graph primitives. The fourth layer is a query-execution layer in which a lightweight controller selects direct lookup, local subgraph retrieval, or macro-node retrieval before grounded generation. This subsection is intended only to provide the architectural overview; the technical details of each layer are presented in the following subsections.

\subsection{Structure-Preserving Multimodal Compilation}
Module 4.1 in Fig.~\ref{fig:overall} produces three compiled outputs for each document: hierarchical paragraph primitives \(\mathcal{T}_i\), conditional-logic table subgraphs \(\mathcal{C}_i\), and formula graph primitives \(\mathcal{F}_i\). Each extracted unit keeps clause identifiers and page coordinates so that every retrieved node can be traced back to the source clause. Table~\ref{tab:primitives} summarizes the primitive types materialized by the compiler.

\begin{table*}[t]
\caption{Typed graph primitives used in SEM-RAG.}
\label{tab:primitives}
\centering
\large
\begin{tabular}{@{}lll@{}}
\toprule
\textbf{Primitive} & \textbf{Main nodes} & \textbf{Key relations} \\
\midrule
Text span & section, paragraph, term & contains, refers-to, defines \\
Table & row header, column header, cell, predicate & row-bind, col-bind, activates \\
Formula & operator, variable, constant & operand-of, precedes, defines \\
\bottomrule
\end{tabular}
\end{table*}

\textit{Hierarchical paragraph extraction (\(\mathcal{T}_i\)):} Each document is first segmented into sections, paragraphs, tables, and displayed equations together with page identifiers and bounding boxes. Paragraph nodes are linked to their enclosing sections, referenced tables, referenced equations, and cross-referenced terms. This part of the compiler preserves the surrounding normative context, which is often where definitions, exceptions, and applicability clauses reside.

\textit{Table-to-graph conversion (\(\mathcal{C}_i\)):} We resolve merged cells and spanning headers into explicit header paths. Each non-empty data cell becomes a value node with incoming edges from its row-header chain, column-header chain, and any footnote or applicability predicate. Units, release tags, and clause numbers are attached as attributes. This representation preserves the selection logic needed for telecom tables such as capability flags, service bitmaps, and power-control matrices.

\textit{Formula-to-graph mapping (\(\mathcal{F}_i\)):} Each equation is normalized, parsed into an abstract syntax tree, and converted into an operator graph. Operators become internal nodes, while variables and constants become leaves. Operand order and precedence are encoded on typed edges. Symbols are then linked to the nearest textual definitions in the same clause or in the preceding definitions block. This allows the retriever to surface either the full equation or the relevant subexpression.

The union of \(\mathcal{T}_i\), \(\mathcal{C}_i\), and \(\mathcal{F}_i\) over all documents forms the base multimodal graph \(\mathcal{G}_0\). In this paper we stop at text, tables, and formulas; rich diagram understanding is left to future work.

\subsection{Entropy-Guided Hierarchical Compression}
After Module 4.1, all compiled primitives are merged into \(\mathcal{G}_0\). Module 4.2 then builds a retrieval hierarchy by minimizing structural entropy. The goal is not to discover semantic topics by prompting an LLM; it is to obtain a compact index that preserves dense local structure and reduces the search space at query time.

For the base graph \(\mathcal{G}_0\), the one-dimensional structural entropy is
\begin{equation}
    H^1(\mathcal{G}) = - \sum_{i=1}^{|\mathcal{V}|} \frac{d_i}{V_G} \log_2 \frac{d_i}{V_G},
\end{equation}
where \(d_i\) is the degree of node \(v_i\) and \(V_G\) is the graph volume \cite{li2016structural}. Given a partition tree \(\mathcal{P} = \{C_1, C_2, \dots, C_k\}\), the corresponding two-dimensional structural entropy is
\begin{equation}
\begin{aligned}
    H^2(\mathcal{G}, \mathcal{P}) = {} & \sum_{j=1}^{k} \frac{V_{C_j}}{V_G} \left(- \sum_{v_i \in C_j} \frac{d_i}{V_{C_j}} \log_2 \frac{d_i}{V_{C_j}}\right) \\
    & - \sum_{j=1}^{k} \frac{g_j}{V_G} \log_2 \frac{V_{C_j}}{V_G},
\end{aligned}
\end{equation}
where \(V_{C_j}\) is the volume of community \(C_j\) and \(g_j\) is the number of cut edges between \(C_j\) and the rest of the graph.

Equation (9) measures the baseline uncertainty of the base graph, whereas Eq. (10) defines the partition objective actually minimized by SEM.

Starting from singleton communities, we iteratively merge adjacent communities with the largest decrease in \(H^2\). In our implementation, local entropy changes are updated with a priority queue, which keeps the procedure efficient on sparse document graphs. No LLM calls are used during this merging stage. Once the hierarchy \(\mathcal{P}^*\) is fixed, we generate short natural-language labels only for a small set of top-level communities.

\subsection{Retrieval-Space Alignment}
Even after compilation, natural-language queries and graph primitives live in different representational spaces. Module 4.3 therefore aligns structural embeddings and textual embeddings of the same primitive. Let \(z_v\) denote the topological feature of node \(v\) obtained from a lightweight message-passing GNN, and let \(t_v\) denote the textual feature of the same node obtained from a pretrained encoder. The alignment goal is to maximize the agreement between the two views,
\begin{equation}
    \max \sum_{v \in \mathcal{V}} I(z_v ; t_v),
\end{equation}
where \(I(\cdot ; \cdot)\) denotes mutual information.

In practice, we do not optimize Eq. (11) directly. Instead, we minimize the following JSD-based surrogate
\begin{equation}
    \mathcal{L}_{\text{align}} = \mathrm{JSD}\!\left(p_g(\cdot \mid v), p_t(\cdot \mid v)\right),
\end{equation}
where \(p_g(\cdot \mid v)\) and \(p_t(\cdot \mid v)\) are projected graph-view and text-view distributions for node \(v\) in a shared space. During training, matched pairs are pulled together while mismatched pairs are pushed apart through negative sampling.

Equation (11) states the ideal agreement objective, and Eq. (12) gives the surrogate used in training.

Within SEM-RAG, this layer serves as a retrieval support module rather than as an independent representation-learning claim. Its role is to reduce the mismatch between natural-language questions and compiled table or formula nodes at query time.

\subsection{Query-Time Controller and Evidence Assembly}
At query time, a lightweight MLP complexity router selects among direct lookup, local subgraph expansion, and macro-node retrieval. The router uses inexpensive query features, including the number of detected entities, the presence of symbolic tokens, query length, and the entropy of the initial vector hits. It is trained on development-set route labels that reflect the minimal retrieval mode needed to answer a question. In the notation of Fig.~\ref{fig:overall}, the router assigns one of three complexity levels: Low, Med, or High.

\begin{enumerate}
    \item \textbf{Low / Factoid mode:} For isolated parameter queries, the system retrieves paragraph and value nodes directly from the vector index.
    \item \textbf{Med / Relational mode:} For dependency or comparison questions, the system anchors entities and expands a bounded \(K\)-hop typed subgraph.
    \item \textbf{High / Global mode:} For abstract or protocol-level questions, the system retrieves a small number of macro-nodes from \(\mathcal{P}^*\) and their summaries.
\end{enumerate}

The retrieved nodes are first assembled into a candidate subgraph \(G_q\). Rather than passing raw graph dumps to the generator, we verbalize \(G_q\) into schema-preserving evidence records \(R_q\) of the form \((\text{clause}, \text{subject}, \text{relation}, \text{object}, \text{condition}, \text{provenance})\). This keeps the LLM context compact while preserving where each retrieved fact came from.

The end-to-end pipeline is summarized in Algorithm~\ref{alg:sem_rag}.

\begin{algorithm}[t]
\caption{SEM-RAG Pipeline}
\label{alg:sem_rag}
\begin{algorithmic}[1]
\REQUIRE Telecom corpus \(\mathcal{D} = \{D_1, D_2, \dots, D_N\}\), user query \(q\)
\ENSURE Grounded answer \(a\)
\STATE \textbf{Offline stage: structure-preserving compilation}
\FOR{each document \(D_i \in \mathcal{D}\)}
    \STATE Detect sections, paragraphs, tables, and equations with provenance metadata
    \STATE Extract hierarchical paragraph primitives \(\mathcal{T}_i\)
    \STATE Compile tables into conditional-logic subgraphs \(\mathcal{C}_i\)
    \STATE Compile formulas into graph primitives \(\mathcal{F}_i\) and link symbols to local definitions
    \STATE Link \(\mathcal{T}_i\), \(\mathcal{C}_i\), and \(\mathcal{F}_i\) into \(\mathcal{G}_0\)
\ENDFOR
\STATE \textbf{Offline stage: hierarchy construction}
\STATE Initialize each node in \(\mathcal{G}_0\) as a singleton community
\STATE Initialize partition tree \(\mathcal{P}\)
\WHILE{a local merge decreases \(H^2(\mathcal{G}_0, \mathcal{P})\)}
    \STATE Merge the adjacent community pair with the largest entropy decrease
\ENDWHILE
\STATE Obtain the hierarchy \(\mathcal{P}^*\) and label only the top-level communities
\STATE \textbf{Offline stage: alignment training}
\STATE Learn aligned graph/text representations for linked primitives using \(\mathcal{L}_{\text{align}}\)
\STATE \textbf{Online stage: routed retrieval}
\STATE Predict complexity level \(\ell \in \{\text{Low}, \text{Med}, \text{High}\}\) with the MLP complexity router
\IF{\(\ell\) is Low}
    \STATE Retrieve paragraph and value nodes by vector search
\ELSIF{\(\ell\) is Med}
    \STATE Anchor entities and expand a bounded \(K\)-hop typed subgraph
\ELSE
    \STATE Retrieve the most relevant macro-nodes from \(\mathcal{P}^*\)
\ENDIF
    \STATE Assemble candidate subgraph \(G_q\) and verbalize it into evidence records \(R_q\)
\STATE Generate \(a \leftarrow \mathrm{LLM}(q, R_q)\)
\RETURN \(a\)
\end{algorithmic}
\end{algorithm}

\section{Complexity Analysis and Scope of Claims}
\label{sec:analysis}
We intentionally avoid theorem-level guarantees on end-to-end answer correctness. The final behavior of the system depends on upstream parsing quality, the sparsity of the compiled graph, the distribution of query types, and the properties of the backbone LLM. What we can state more cleanly are indexing and retrieval cost implications under explicit assumptions.

We assume that: 1) the compiled document graph is sparse, that is, \(|\mathcal{E}| = O(|\mathcal{V}|)\), which is typical when each text span, table cell, and operator introduces only a bounded number of local edges; 2) after SEM compression, at most \(k\) top-level communities are sent to the LLM for labeling, with \(k\) fixed by a summary budget; 3) local online retrieval expands at most \(K_{\max}\) hops and global retrieval returns at most \(M\) macro-nodes; and 4) the vector index supports sublinear nearest-neighbor lookup.

\begin{proposition}[Indexing-time LLM token usage]
Under Assumptions 1 and 2, the LLM token usage of SEM-RAG during hierarchy construction is \(O(kT_s)\), where \(T_s\) is the maximum token budget for one top-level summary. By contrast, an LLM-centered hierarchical pipeline that summarizes every community across \(L\) levels requires
\begin{equation}
    \Omega\!\left(T_s \sum_{l=1}^{L} |C_l|\right)
\end{equation}
tokens, where \(|C_l|\) is the number of communities at level \(l\).
\end{proposition}

\textit{Proof:} In SEM-RAG, SEM merging operates directly on the graph and does not call the LLM. LLM usage is deferred until the hierarchy is fixed, after which only \(k\) top-level communities are labeled, each with a budget of at most \(T_s\) tokens. The total token usage is therefore bounded by \(kT_s\). In a baseline that summarizes every community across multiple levels, each community incurs one summary call of size at most \(T_s\), giving a lower bound proportional to \(T_s \sum_l |C_l|\). \hfill \(\square\)

\begin{remark}[Expected query-time retrieval cost]
Under Assumptions 3 and 4, the expected retrieval cost can be expressed as
\begin{equation}
\begin{aligned}
    \mathbb{E}[T_{\text{ret}}] = {} & p_{\mathrm{low}} O(\log |\mathcal{V}|) \\
    & + p_{\mathrm{med}} O(|A_q| d_{\mathrm{eff}}^{K_{\max}}) \\
    & + p_{\mathrm{high}} O(M),
\end{aligned}
\end{equation}
where \(p_{\mathrm{low}}, p_{\mathrm{med}}, p_{\mathrm{high}}\) denote the query-type proportions, \(A_q\) is the set of query anchors, and \(d_{\mathrm{eff}}\) is the effective degree after compression. If the router has accuracy \(\alpha < 1\), misrouted queries add an overhead term \(\Delta_{\alpha}\) corresponding to the expected gap between the chosen route and the minimal required route. This is not a worst-case end-to-end latency guarantee; it only clarifies that routed retrieval confines online work to bounded retrieval modes instead of unconstrained traversal over the full graph.
\end{remark}

\section{Experimental Design and Evaluation}
\label{sec:experiments}
We evaluate SEM-RAG with a deliberately focused question in mind: does structure-preserving document compilation improve retrieval over telecom standards, and can SEM make that richer representation practical at indexing time? This framing follows the paper's narrowed thesis. The central claims concern compiled graph primitives and entropy-guided compression. By contrast, the alignment layer and the query-time controller are evaluated as supporting mechanisms that help the compiled index operate efficiently online rather than as isolated scientific contributions in their own right. Accordingly, the experiments are organized to test each link in this chain instead of reporting only aggregate question-answering scores.

Four questions guide the study. First, does compilation into table and formula primitives reduce the failures caused by structural amnesia? Second, once the corpus is compiled into a graph, can SEM compress it without damaging the multi-hop dependencies that telecom queries require? Third, do alignment and routing improve online efficiency without becoming the main source of the observed accuracy gains? Fourth, are the improvements stable across backbone LLMs with different pretraining and deployment regimes? The following subsections address these questions in the same order.

\subsection{Experimental Setup and Implementation Details}
\textbf{Evaluation protocol:} All methods index the same source collections and answer the same held-out benchmark questions. Unless a baseline requires a different internal policy, we keep the generator configuration fixed and cap the final evidence budget at five retrieved items so that answer quality is compared under a similar context budget. For passage-based baselines, an evidence item is a retrieved text chunk or summary node. For SEM-RAG, an evidence item after verbalization may correspond to a paragraph node, a compiled table value with its row and column header path, an operator subgraph, or a macro-node summary. This distinction is important: the comparison is not between larger and smaller prompts, but between different ways of representing the same underlying evidence.

\textbf{Infrastructure and hyperparameters:} SEM-RAG is implemented in Python 3.10 and PyTorch 2.2. Offline parsing and SEM optimization run on a dual-socket Intel Xeon Platinum 8380 server with 512GB RAM. For dense embeddings, we use OpenAI's \texttt{text-embedding-3-large} model (dimension \(3072\)). Unless otherwise stated, generator calls use temperature \(0.1\) and \texttt{top\_p} \(= 0.9\) to keep sampling variance low. The reported indexing times include document normalization, multimodal graph construction, SEM merging, and the final top-level summarization stage, while excluding one-time environment setup and dataset download. Query latency is measured under warm-start conditions with batch size one, from query arrival to final answer generation.

\textbf{Parser configuration:} Documents are first converted into a structured intermediate JSON containing reading order, clause identifiers, table boundaries, equation spans, and page coordinates. Merged row and column headers are resolved into explicit header paths rather than left as visual spans. Each non-empty table cell is stored together with its row-header path, column-header path, local predicates, unit strings, and provenance coordinates. Equations are normalized and parsed into abstract syntax trees before operator-graph construction, after which symbols are linked back to nearby textual definitions whenever available. In other words, the parser does not merely preserve layout for display; it converts layout into executable retrieval units that can later be verbalized as clause-aware evidence records. This detail matters because the main gains of SEM-RAG depend on the quality of the compiled graph, not only on downstream retrieval heuristics.

\textbf{Baselines and retrieval settings:} We benchmark SEM-RAG against four representative architectures:
\begin{itemize}
    \item \textbf{Naive RAG:} A recursive character text splitter (chunk size \(1024\), overlap \(128\)) with Maximum Inner Product Search (MIPS) over a Milvus vector database.
    \item \textbf{RAPTOR:} A hierarchical text-tree retrieval system based on recursive summarization.
    \item \textbf{LightRAG~\cite{guo2024lightrag}:} A dual-level entity retrieval framework optimized for operational speed.
    \item \textbf{Standard GraphRAG~\cite{edge2024global}:} A graph baseline that uses LLMs for entity extraction and bottom-up community summarization.
\end{itemize}
Unless specifically analyzed in Section~\ref{subsec:llm_stability}, \textbf{GPT-5.2} serves as the primary generative backbone for all evaluated systems. For SEM-RAG, the controller selects among direct lookup, bounded subgraph expansion (up to \(K=3\) hops), and macro-node retrieval (up to five top-level nodes). This setup is meant to reflect the intended role of each module: compiled structure and SEM determine what knowledge is available and how it is indexed, whereas routing only decides how much of that index must be activated for a particular query.

\textbf{Metrics, subsets, and annotation protocol:} We report task accuracy, MRR@10, EM/F1, indexing token usage, indexing wall-clock time, and query latency. Indexing time is averaged over three runs. Query latency is averaged over 1,000 held-out questions. For retrieval evaluation, MRR@10 is computed from the rank of the first retrieved item whose provenance overlaps the gold supporting evidence within the top 10 results; the gold supporting evidence is defined as the minimal clause, table cell, or formula span annotated for evaluation. EM is exact string match after normalization. F1 is token-level answer F1 after lowercasing, punctuation removal, and acronym canonicalization. In Table~\ref{tab:overall_performance}, the reported F1 is macro-averaged over the TeleQnA and ORAN-Bench-13K answer sets. To make the analysis more diagnostic, we construct targeted subsets in addition to the original benchmark splits. A question is labeled \textit{table-heavy} when the reference answer requires an explicit row and column binding, a table footnote, or a condition attached to a specific cell. A question is labeled \textit{formula-heavy} when the reference answer depends on operator structure, symbol grounding, or a numerical constraint that cannot be recovered from surrounding prose alone. For the multi-hop analysis, hop depth is assigned according to the shortest supporting typed path in the compiled graph and then checked against the cited clauses. For the error analysis in Figure~\ref{fig:qa_hallucination}b, we manually inspect 200 randomly sampled erroneous answers per system and assign each failure to row or column mismatch, mathematical inconsistency, or unsupported synthesis. These targeted views are important because average QA scores alone can hide whether a method truly solves structural amnesia or simply benefits from better generic retrieval.

\textbf{Benchmark datasets:} Our evaluation uses a set of expert-curated telecommunications datasets:
\begin{itemize}
    \item \textbf{SPEC5G~\cite{karim2023spec5g}:} A 5G protocol security dataset extracted directly from cellular network specifications.
    \item \textbf{TSpec-LLM~\cite{nikbakht2024tspec}:} A multimodal dataset spanning 3GPP Releases 8 to 19 and retaining technical tables, parameter grids, and mathematical derivations.
    \item \textbf{TeleQnA~\cite{maatouk2023teleqna}:} A 10,000-question benchmark for inferential telecom knowledge.
    \item \textbf{ORAN-Bench-13K~\cite{gajjar2024oranbench}:} An acronym-dense benchmark for disaggregated Open RAN architectures.
\end{itemize}
Together, these datasets cover both broad telecom knowledge and the more brittle cases that depend on standards-specific layouts, such as service bitmaps, capability matrices, and formula-derived thresholds. We do not force every analysis to use every benchmark, because the diagnostics target different properties. End-to-end QA is reported on TeleQnA and ORAN-Bench-13K because they provide broad question coverage and stable answer formats. Structure-sensitive diagnostics such as table-heavy, formula-heavy, and multi-hop evaluation are drawn from TeleQnA and TSpec-LLM, where supporting evidence can be traced back to specific clauses, table cells, or equations. Scalability and latency are systems measurements on the indexed corpus and are therefore reported separately from benchmark accuracy.

\subsection{Multimodal QA Performance and Hallucination Mitigation}
This subsection directly tests the paper's primary claim: in telecom standards, much of the failure of conventional RAG comes from losing structural relations during ingestion rather than from a lack of semantic similarity alone. Unless otherwise noted, Table~\ref{tab:overall_performance} reports full-benchmark end-to-end QA results on TeleQnA and ORAN-Bench-13K, whereas Figure~\ref{fig:qa_hallucination} reports table-heavy and formula-heavy subsets drawn from TeleQnA and TSpec-LLM. To make the structural effect visible, we do not rely only on benchmark-wide averages. Instead, we isolate questions whose gold evidence necessarily passes through a compiled table primitive or a compiled formula primitive. If SEM-RAG is helping mainly because it preserves executable structure, then its advantage should become especially large on these subsets.

\begin{figure}[t]
    \centering
    \includegraphics[width=\columnwidth]{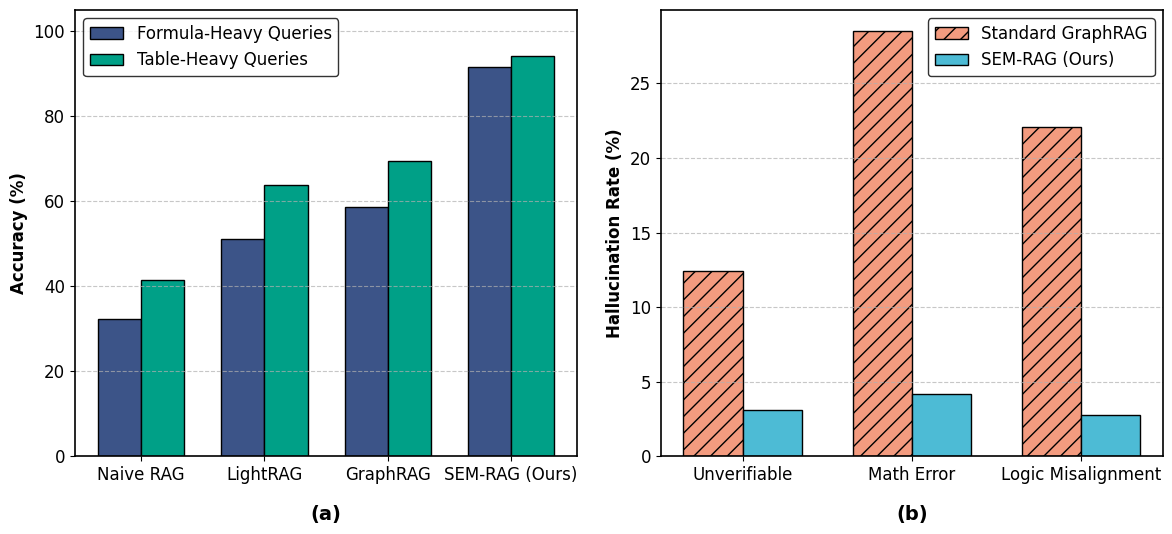}
    \caption{Performance on table-heavy and formula-heavy questions, together with manually annotated error categories. SEM-RAG gains are concentrated on cases where preserving header bindings or operator structure matters.}
    \label{fig:qa_hallucination}
\end{figure}

Table~\ref{tab:overall_performance} first shows that the advantage is not restricted to a narrow slice of the benchmarks. SEM-RAG reaches 94.1\% accuracy on TeleQnA, 93.8\% on ORAN-Bench-13K, and the best MRR@10 and overall F1 among all compared systems. The improvement over Standard GraphRAG is already large at the aggregate level, which indicates that the compiled representation remains useful even when the benchmark mixes simple acronym questions with more complex procedural ones.

\begin{table}[t]
\caption{Overall performance on TeleQnA and ORAN-Bench-13K with GPT-5.2}
\label{tab:overall_performance}
\centering
\begin{tabular}{@{}lcccc@{}}
\toprule
\textbf{Architecture} & \textbf{TeleQnA} & \textbf{ORAN-Bench} & \textbf{MRR@10} & \textbf{F1} \\
\midrule
Naive RAG & 45.5\% & 42.1\% & 0.38 & 0.44 \\
RAPTOR & 74.2\% & 68.5\% & 0.67 & 0.70 \\
LightRAG~\cite{guo2024lightrag} & 76.8\% & 75.6\% & 0.69 & 0.74 \\
Standard GraphRAG~\cite{edge2024global} & 81.2\% & 78.4\% & 0.72 & 0.79 \\
\textbf{SEM-RAG (Ours)} & \textbf{94.1\%} & \textbf{93.8\%} & \textbf{0.89} & \textbf{0.92} \\
\bottomrule
\end{tabular}
\end{table}

The subset analysis in Figure~\ref{fig:qa_hallucination}a is even more revealing. On formula-heavy questions, text-first baselines degrade sharply because equations are retrieved either as broken strings or as nearby prose that mentions only part of the computation. LightRAG and Standard GraphRAG reach 51.2\% and 58.7\% accuracy on this subset, whereas SEM-RAG reaches 91.5\%. The gain here is consistent with the design of the formula compiler: operator nodes preserve precedence and operand structure, while symbol links preserve what each term actually means in the clause where it is defined. This makes the retrieved evidence closer to an executable expression than to a quoted line of math-marked text.

A similar pattern appears on table-heavy queries. In telecom standards, the answer is often not the literal text around a table but the correct binding between a row header, a column header, a footnote condition, and a cell value. Passage retrieval can bring the right page yet still return the wrong cell because the row and column logic has been flattened away. SEM-RAG performs better precisely because the compiled graph stores those bindings explicitly. In that sense, the parser is not a cosmetic preprocessing step; it determines whether the retrieval stage can recover the correct operational constraint at all.

The manual error analysis in Figure~\ref{fig:qa_hallucination}b supports this interpretation. Standard GraphRAG shows a 28.5\% mathematical inconsistency rate and a 22.1\% row or column logic error rate within the annotated error sample. SEM-RAG reduces row or column logic errors to 2.8\%, which is consistent with the fact that header paths, predicates, and cell values remain attached inside the graph. Importantly, this reduction is not explained by larger context windows or heavier prompting; the evidence budget is matched across systems. The result therefore supports the paper's main thesis that better structure, rather than only better semantic recall, is the primary driver of factual gains on these telecom tasks.

\subsection{Multi-hop Reasoning and Depth Traversability}
The previous subsection shows that compiled primitives help local factual grounding. The next question is whether SEM compression preserves enough connectivity for longer reasoning chains. This is a critical test for the second core claim of the paper. A compressed hierarchy is only useful if it reduces indexing cost without erasing the relational paths needed to answer cross-clause telecom questions. For this diagnostic, we use the subset of TeleQnA and TSpec-LLM questions for which a shortest supporting typed path can be annotated in the compiled graph.

To examine this, we evaluate subsets of questions that require 1-hop, 2-hop, and 3-hop evidence chains across the compiled graph. Here, a hop corresponds to a typed transition such as paragraph-to-term, term-to-table-cell, or symbol-to-definition. This categorization makes the analysis more informative than a generic QA score because it distinguishes direct parameter lookup from dependency tracing across multiple compiled primitives.

\begin{table}[ht]
\caption{Accuracy (\%) versus multi-hop reasoning depth}
\label{tab:multihop}
\centering
\begin{tabular}{@{}lccc@{}}
\toprule
\textbf{Model architecture} & \textbf{1-Hop} & \textbf{2-Hop} & \textbf{3-Hop} \\ \midrule
Naive RAG & 68.2 & 34.1 & 12.5 \\
RAPTOR & 76.5 & 48.2 & 22.4 \\
LightRAG & 80.1 & 62.5 & 45.3 \\
Standard GraphRAG & 82.5 & 71.4 & 58.7 \\ \midrule
\textbf{SEM-RAG (Ours)} & \textbf{96.2} & \textbf{92.5} & \textbf{87.1} \\ \bottomrule
\end{tabular}
\end{table}

Table~\ref{tab:multihop} shows that all methods remain reasonably competitive on 1-hop lookup, where strong lexical overlap already helps. The gap widens quickly as reasoning depth increases. At 3 hops, Standard GraphRAG reaches 58.7\% while SEM-RAG reaches 87.1\%. This behavior is significant for two reasons. First, it suggests that SEM compression is not merely shrinking the graph; it is retaining the dense local neighborhoods that carry the real explanatory structure of the specification. Second, it shows that the query does not lose contact with that structure as traversal proceeds.

The role of the alignment layer is best understood in this context. JSD alignment is not the main reason these paths exist; the paths are created by compilation. Its value is that it helps the query anchor onto the right neighborhood so that multi-hop traversal does not drift toward semantically similar but structurally irrelevant nodes. The result is therefore consistent with the paper's division of labor: compilation and SEM define the usable graph, while alignment helps the retriever stay on the right path within that graph.

\subsection{Indexing Scalability}
This subsection evaluates the practical consequence of replacing LLM-driven hierarchical indexing with SEM-based compression. For telecom standards, this issue is not secondary. Standards evolve across releases and change requests, so a retrieval system that is accurate but prohibitively expensive to rebuild is difficult to maintain in practice. The relevant question is therefore not only whether the hierarchy is useful, but whether it can be rebuilt without an LLM-dependent bottleneck. This is a benchmark-agnostic systems test conducted on progressively larger slices of the indexed telecom corpus.

Consistent with Proposition~1 in Section~\ref{sec:analysis}, we measure how indexing-time cost evolves with corpus size. For Standard GraphRAG, the reported cost includes entity extraction, graph construction, recursive community processing, and the associated LLM summaries. For SEM-RAG, the reported cost includes parsing, graph construction, entropy-guided merging, and only the final small set of top-level summaries. The comparison should therefore be read as a study of LLM-dependent indexing cost, not as a claim that every part of the system scales identically.

\begin{table}[ht]
\caption{Hierarchical indexing scalability: cost and time versus corpus size}
\label{tab:scalability}
\centering
\begin{tabular}{@{}c|cc|cc@{}}
\toprule
\multirow{2}{*}{\textbf{\begin{tabular}[c]{@{}c@{}}Corpus size\\ (entities)\end{tabular}}} & \multicolumn{2}{c|}{\textbf{Standard GraphRAG}} & \multicolumn{2}{c}{\textbf{SEM-RAG (Ours)}} \\ \cmidrule(l){2-5}
& \textbf{\begin{tabular}[c]{@{}c@{}}Tokens\\ (Millions)\end{tabular}} & \textbf{\begin{tabular}[c]{@{}c@{}}Time\\ (Hours)\end{tabular}} & \textbf{\begin{tabular}[c]{@{}c@{}}Tokens\\ (Millions)\end{tabular}} & \textbf{\begin{tabular}[c]{@{}c@{}}Time\\ (Hours)\end{tabular}} \\ \midrule
10,000 & 0.5 & 0.4 & \textbf{0.05} & \textbf{0.02} \\
50,000 & 3.2 & 2.1 & \textbf{0.08} & \textbf{0.05} \\
100,000 & 8.5 & 6.4 & \textbf{0.12} & \textbf{0.09} \\
500,000 & 45.0 & 31.5 & \textbf{0.25} & \textbf{0.32} \\
1,000,000 & 110.0 & 78.2 & \textbf{0.35} & \textbf{0.68} \\ \bottomrule
\end{tabular}
\end{table}

\begin{figure}[t]
    \centering
    \includegraphics[width=\columnwidth]{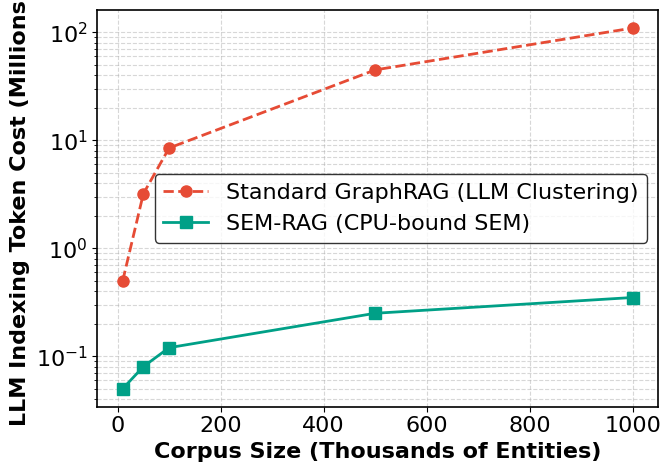}
    \caption{Indexing-time scalability comparison. SEM shifts hierarchy construction from LLM calls to graph optimization, which keeps the token budget small as corpus size grows.}
    \label{fig:scalability}
\end{figure}

Table~\ref{tab:scalability} and Figure~\ref{fig:scalability} show a clear divergence as the corpus grows. At one million entities, Standard GraphRAG consumes 110 million tokens and about 78 hours of processing. SEM-RAG uses 0.35 million tokens and less than one hour because the expensive LLM stage is postponed until after the hierarchy is fixed and then applied only to a small number of top-level communities. The practical interpretation is straightforward: the richer document representation introduced by SEM-RAG becomes viable precisely because hierarchy construction no longer requires repeated semantic summarization at every intermediate level.

This result also clarifies the scope of the claim. We are not arguing that graph construction becomes free, or that SEM is asymptotically optimal in every implementation. The empirical point is narrower and more defensible: once hierarchical compression is moved from repeated LLM calls to graph-topological optimization, the LLM-dependent portion of indexing no longer scales with the full number of intermediate communities. That is the systems property that makes the proposed compiled representation more realistic for continuously updated standards corpora.

\subsection{End-to-End Latency Breakdown}
While indexing is an offline concern, a telecom assistant still has to answer online queries with low delay. This subsection therefore evaluates the alignment layer and the controller in the role they are meant to play in the paper: as systems mechanisms that keep retrieval practical at serving time. The question is not whether they create the main performance gain, but whether they preserve that gain without introducing a prohibitive online cost. Latency is measured on the same 1,000-query held-out pool sampled from TeleQnA and ORAN-Bench-13K.

\begin{figure}[t]
    \centering
    \includegraphics[width=\columnwidth]{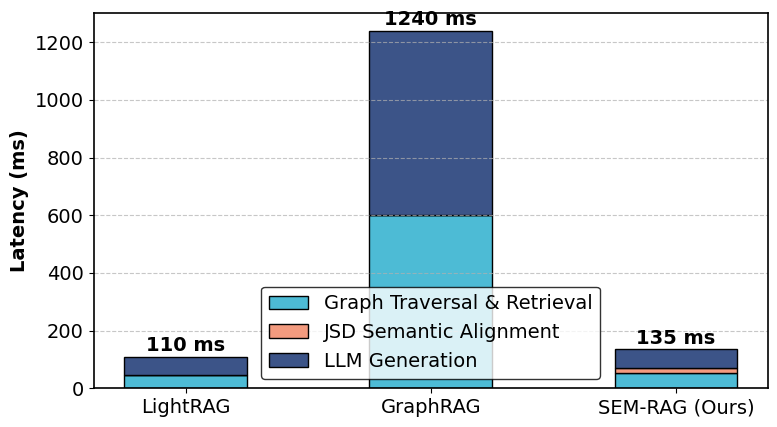}
    \caption{Breakdown of online latency into retrieval, alignment, and generation components. The alignment layer adds modest overhead, while routed retrieval keeps total latency close to lightweight baselines.}
    \label{fig:latency_breakdown}
\end{figure}

Figure~\ref{fig:latency_breakdown} decomposes online time into retrieval, alignment, and generation. SEM-RAG spends about \(18\) ms on alignment, but this overhead is small relative to the savings obtained from bounded retrieval and smaller evidence sets. The total online latency remains around \(135\) ms, close to LightRAG and well below Standard GraphRAG, whose retrieval stage exceeds \(1200\) ms in our setup. This indicates that the extra structure introduced at indexing time does not force the system into an expensive online traversal regime.

The result is easiest to interpret through the router. Many benchmark questions are local enough that full graph traversal is unnecessary. By routing these cases directly to vector lookup, the system avoids paying a graph-expansion cost on every query. More complex questions still trigger subgraph expansion or macro-node retrieval, but only when the query actually demands it. In other words, the controller does not create the knowledge used by the answer; it keeps the online cost of accessing that knowledge bounded.

It is also worth noting that generation remains the dominant fraction of end-to-end latency. This is why the online gains are visibly smaller than the indexing-time gains reported in Section~\ref{sec:experiments}. The alignment layer and the controller should therefore be read as enabling components: they preserve most of the retrieval benefit of the compiled graph without turning the serving path into a latency bottleneck.

\subsection{Ablation Study on Structural Components}
The ablation study isolates which parts of SEM-RAG are responsible for the final gains. To keep the comparison controlled, we hold the document pool, backbone generator, and evidence budget fixed, and remove one component at a time from the full system. We report ablations on TeleQnA, TSpec-LLM, and ORAN-Bench-13K because these three datasets stress formula reasoning, multimodal grounding, and table-intensive protocol lookup, respectively. We intentionally do not include a ``w/o SEM'' row in this table because removing SEM changes the hierarchy itself and therefore the entire indexing regime; its effect is isolated more cleanly in the scalability and multi-hop analyses above.

\begin{table}[ht]
\caption{Ablation study of core components across multiple datasets}
\label{tab:ablation}
\centering
\resizebox{\columnwidth}{!}{%
\begin{tabular}{@{}lccc@{}}
\toprule
\textbf{Architecture variant} & \textbf{\begin{tabular}[c]{@{}c@{}}TeleQnA\\ (EM \%)\end{tabular}} & \textbf{\begin{tabular}[c]{@{}c@{}}TSpec-LLM\\ (F1 \%)\end{tabular}} & \textbf{\begin{tabular}[c]{@{}c@{}}ORAN-Bench\\ (Acc. \%)\end{tabular}} \\ \midrule
\textbf{Full SEM-RAG pipeline} & \textbf{94.1} & \textbf{91.2} & \textbf{93.8} \\ \midrule
w/o Table-to-Graph & 58.2 & 61.4 & 55.7 \\
w/o Formula-to-OPG & 63.5 & 68.7 & 69.1 \\
w/o JSD Alignment & 71.8 & 73.5 & 75.2 \\
w/o MLP Routing & 82.4 & 84.1 & 81.9 \\ \bottomrule
\end{tabular}%
}
\end{table}

\begin{figure}[t]
    \centering
    \includegraphics[width=\columnwidth]{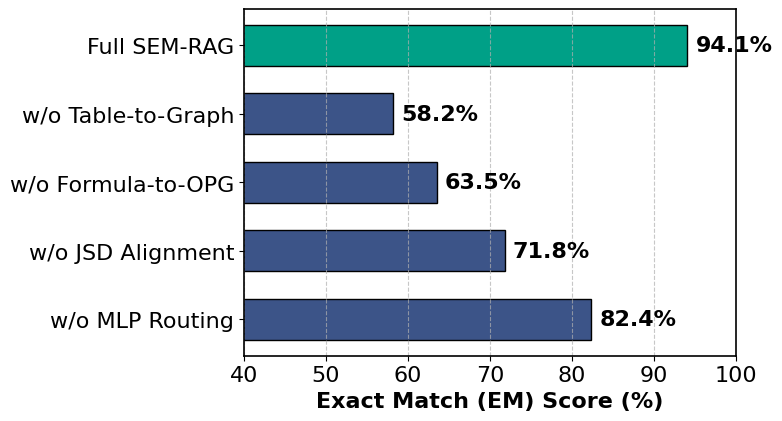}
    \caption{Effect of removing core SEM-RAG components on held-out QA performance.}
    \label{fig:ablation}
\end{figure} 

The ranking of the ablations matches the paper's intended thesis. Removing the table compiler or the formula compiler causes the largest degradation, whereas removing alignment or routing still hurts performance but to a lesser extent. This pattern matters because it shows that the central gains come from compiled structure rather than from query-time heuristics.

\textbf{Table-to-Graph conversion:} Removing the table compiler hurts ORAN-Bench the most, dropping accuracy from 93.8\% to 55.7\%. This is consistent with the heavy use of matrices, capability tables, and state-dependent profiles in Open RAN documents. Once the row and column binding is flattened back into text, the retriever may still land near the relevant table but can no longer reliably identify the correct cell under the applicable condition.

\textbf{Formula-to-OPG mapping:} Removing the operator-graph compiler reduces TeleQnA exact match to 63.5\%, showing that formula understanding does not transfer reliably when equations are left as flat strings. In these cases, the generator often sees the right symbols but not the executable structure that determines precedence, substitution, or constraint scope. The drop therefore supports the claim that formula compilation is necessary for stable quantitative reasoning over telecom clauses.

\textbf{JSD alignment:} Disabling alignment lowers TSpec-LLM F1 to 73.5\%. This confirms that alignment is useful, but the interpretation should remain modest. The graph still exists and still contains the correct compiled evidence; the main issue is that the query anchors less precisely to the intended neighborhood. This is exactly the role assigned to alignment in the paper: it is an enabling retrieval layer that protects the benefits of the compiled graph at query time.

\textbf{MLP routing:} Forcing the system to always retrieve from larger graph regions reduces TeleQnA EM to 82.4\%. More evidence does not automatically mean better evidence. Without routing, the generator receives larger but noisier contexts, especially for short factoid queries whose answer was already available through direct lookup. The router should therefore be read as a systems optimization that controls noise and latency, not as the principal source of semantic improvement.

In general, the ablation results reinforce a clear hierarchy among the components. The parser and the compiled primitives are foundational, SEM makes them scalable, and alignment plus routing help those benefits survive in the online pipeline.

\subsection{Architecture Generalizability Across Foundation Models}
\label{subsec:llm_stability}
Finally, we test whether the reported gains are tied to one generator family or whether they persist across models with different pretraining and serving characteristics. For this study, we keep the retrieval pipeline fixed and vary only the backbone LLM used for answer synthesis. This isolates whether SEM-RAG mainly improves the evidence supplied to the generator, or whether it depends on model-specific prompt behavior. This experiment uses TeleQnA only, because it offers the broadest question coverage and the most stable answer normalization across generators.

\begin{table}[ht]
\caption{Cross-LLM stability analysis on TeleQnA}
\label{tab:llm_stability}
\centering
\begin{tabular}{@{}lccc@{}}
\toprule
\textbf{Foundation LLM} & \textbf{\begin{tabular}[c]{@{}c@{}}Naive RAG\\ (Acc. \%)\end{tabular}} & \textbf{\begin{tabular}[c]{@{}c@{}}SEM-RAG\\ (Acc. \%)\end{tabular}} & \textbf{\begin{tabular}[c]{@{}c@{}}Absolute\\ gain\end{tabular}} \\ \midrule
GPT-5.2 (Closed) & 45.5 & 94.1 & \textbf{+48.6} \\
Claude 3.5 Sonnet (Closed) & 43.2 & 92.8 & \textbf{+49.6} \\
LLaMA-3-70B (Open) & 38.7 & 89.5 & \textbf{+50.8} \\
Mistral-Large (Open) & 40.1 & 90.2 & \textbf{+50.1} \\ \bottomrule
\end{tabular}
\end{table}

Table~\ref{tab:llm_stability} shows that baseline accuracy varies by model, as expected. However, the absolute gain provided by SEM-RAG stays close to \(50\) percentage points across all four backbones. The same retrieval design therefore benefits both closed and open model families, despite their differences in memorized telecom knowledge and instruction-following behavior.

This stability is important for the interpretation of the paper. If the gains came mainly from prompt sensitivity or hidden domain knowledge in one proprietary model, the improvement would fluctuate much more strongly across backbones. Instead, the uplift remains consistent, which supports the view that SEM-RAG improves the quality and structure of the retrieved evidence itself. At the same time, the final absolute scores are not identical across LLMs, which is expected: once the correct evidence is retrieved, the generator still has to synthesize the answer in a fluent and faithful way.

To sum up, the cross-model results reinforce the paper's main message. SEM-RAG should be understood primarily as a retrieval and indexing design for standards-heavy corpora, not as a prompt recipe tailored to one foundation model.

\section{Discussion and Limitations}
\label{sec:limitations}
The paper makes a focused claim, not a universal one. We do not claim formal guarantees on answer correctness, and the compiled graph can still inherit upstream parsing errors, especially from low-quality PDFs or heavily nested annex tables. The current compiler is designed for text, tables, and formulas; message sequence charts and other highly graphical artifacts are only weakly represented through surrounding text and provenance links. In addition, the complexity analysis in Section~\ref{sec:analysis} concerns indexing and retrieval budgets rather than semantic correctness. Finally, all experiments are conducted on telecom standards, so cross-domain transfer remains an empirical question rather than an established property.

\section{Conclusion}
\label{sec:conclusion}
SEM-RAG addresses a specific but important problem: how to retrieve from telecom standards without destroying table and formula semantics during document ingestion. The proposed system combines structure-preserving compilation with entropy-guided graph compression, then uses alignment and lightweight routing to make the resulting graph practical at query time. Across multiple telecom benchmarks, this design improves performance on table-heavy and formula-heavy questions, retains strong multi-hop reasoning ability, and substantially reduces indexing-time LLM usage relative to standard GraphRAG. In summary, we conclude that in this domain, preserving document structure is not a secondary preprocessing detail but a prerequisite for reliable retrieval. Future work will extend the compiler to sequence charts and richer engineering diagrams, and will study how the same design transfers to other standards-heavy technical domains.

\bibliographystyle{IEEEtran}
\bibliography{ref}

\end{document}